\documentclass[journal]{IEEEtran}

\usepackage{xspace}
\usepackage[overload]{empheq}
\usepackage{cleveref} 
\usepackage{amsmath}
\usepackage{amssymb}
\usepackage{amsthm}
\usepackage{mathptmx}
\usepackage{cite}
\usepackage{setspace}
\usepackage[font={small}]{caption}

\usepackage{dsfont}
\usepackage{color}
\usepackage[lofdepth,lotdepth]{subfig}
\usepackage{stmaryrd}
\usepackage{todonotes}

\DeclarePairedDelimiter\ceil{\lceil}{\rceil}
\DeclarePairedDelimiter\floor{\lfloor}{\rfloor}
\usepackage[linesnumbered,ruled,vlined]{algorithm2e} 

\makeatletter
\newcommand{\removelatexerror}{\let\@latex@error\@gobble}
\makeatother

\DeclareMathOperator*{\argmax}{arg\,max}
\DeclareMathOperator*{\argmin}{arg\,min}

\usepackage{tikz}
\usetikzlibrary{shapes}
\usetikzlibrary{graphs}
\usetikzlibrary{calc}

\usepackage{tkz-berge}

\newtheorem{proposition}{Proposition}
\newtheorem{definition}{Definition}

\newtheorem{example}{Example}
\newtheorem{lemma}{Lemma}

\newtheorem{corrollary}{Corrollary}

\title{A Dynamic Game Approach to Designing Secure Interdependent IoT-Enabled Infrastructure Network}

\author{Juntao Chen, ~\IEEEmembership{Member,~IEEE}, Corinne Touati, and Quanyan Zhu, ~\IEEEmembership{Member,~IEEE} 
\thanks{This paper has been accepted for publication in \textit{IEEE Transactions on Network Science and Engineering}.}
\thanks{This work was supported in part by the National Science of Foundation under Grants ECCS-1847056, CNS-1544782, CNS-2027884, and SES-1541164, and in part by the Army Research Office (ARO) under Grant W911NF-19-1-0041.} 
\thanks{Juntao Chen is with the Department of Computer and Information Sciences, Fordham University,  New York,  NY 10023 USA. E-mail: jchen504@fordham.edu}
\thanks{Corinne Touati is with INRIA, F38330 Montbonnot Saint-Martin, France. E-mail: corinne.touati@inria.fr}
\thanks{Quanyan Zhu is with the Department of Electrical and Computer Engineering, Tandon School of Engineering, New York University, Brooklyn, NY 11201, USA. E-mail: qz494@nyu.edu}
}

\begin{document}

\maketitle

\begin{abstract}
The emerging Internet of Things (IoT) applications that leverage ubiquitous connectivity and big data are facilitating the realization of smart everything initiatives. IoT-enabled infrastructures have naturally a multi-layer system architecture with an overlaid or underlaid device network and its coexisting infrastructure network. The connectivity between different components in these two heterogeneous interdependent networks plays an important role in delivering real-time information and ensuring a high-level situational awareness. However, IoT-enabled infrastructures face cyber threats due to the wireless nature of communications. Therefore, maintaining network connectivity in the presence of adversaries is a critical task for infrastructure network operators. In this paper, we establish a three-player three-stage dynamic game-theoretic framework including two network operators and one attacker to capture the secure design of multi-layer interdependent infrastructure networks by allocating limited resources. We use subgame perfect Nash equilibrium (SPE) to characterize the strategies of players with sequential moves. In addition, we assess the efficiency of the equilibrium network by comparing with its team optimal solution counterparts in which two network operators can coordinate. We further design a scalable algorithm to guide the construction of the equilibrium IoT-enabled infrastructure networks. Finally, we use case studies on the emerging paradigm of the Internet of Battlefield Things (IoBT) to corroborate the obtained results.
\end{abstract}


\section{Introduction}
The massive deployment of Internet of Things (IoT) technologies provides ubiquitous connectivity for heterogeneous machines and devices for data collection, information exchange and operational decision-making. Therefore, IoT is widely adopted in various application domains especially in the infrastructures including smart grids, smart homes, intelligent transportations, and smart cities \cite{gubbi2013internet}. 
With the current information and communication technologies (ICTs), an IoT-enabled infrastructure network has its own networking platform that is interoperable within the existing Internet infrastructure. Hence, an IoT-enabled infrastructure can be naturally viewed as a two-layer interdependent network consisting of the infrastructure layer network and the overlaid or underlaid device layer network. For instance, in the Internet of Battlefield Things (IoBT), the soldier networks equipped with wearable devices are integrated with the unmanned aerial vehicle (UAV) ad hoc networks to perform tasks. The connections in the two-layer network architecture can be classified into two types: (i) the \textit{interlinks} by which devices/infrastructures communicate between themselves as well as (ii) the \textit{intralinks} by which devices communicate with the infrastructure.

The connectivity of an IoT-enabled infrastructure network plays an important role in information dissemination and real-time decision-making for mission-critical operations. Note that devices can communicate with each other or with infrastructures to maintain a global situational awareness of the network. Furthermore, the IoT devices which are scarce of on-board computational resources can outsource heavy computations to the data centers through cloud computing infrastructure \cite{Kumar}. IoT-enabled infrastructures are often vulnerable to cyberattacks which can degrade the system performance, since most of the communications within the IoT networks are wireless in nature. For example, in IoBT networks, the communications between a soldier and a UAV relay node can be jammed by an attacker, and a soldier thus becomes isolated and loses information and awareness of the battlefield.

Therefore, to protect the IoT-enabled infrastructure from adversarial behaviors, it is imperative to design secure and robust two-layer networks that can maintain connectivity despite of link failures. 
Due to the heterogeneous and two-layer feature of the network, the design of the network is decentralized essentially. Specifically, the  network design involves two players who design their own subnetworks sequentially. As in IoBT networks, UAVs form their own relay networks, while a team of soldiers forms a network based on the knowledge of UAV locations to maintain the communications among soldiers and command and control stations. The objectives of these two network operators are to maintain the connectivity of the global network by considering network creation cost, while an attacker aims to disconnect the network at the minimum attacking cost. 

In this paper, we use a three-player three-stage game to capture the secure interdependent IoT-enabled infrastructure network design. At the first stage, the network operator 1 creates links by anticipating the behavior of the network operator 2 and the adversary. At the second stage, the network operator 2 observes the links created by operator 1 and forms links to secure the network by anticipating the adversarial behaviors. Finally, the adversary observes the whole network created by the two operators and launches an attack targeting to disconnect the network. The two operators have aligned objectives to make the two-layer network connected. However, they have different costs or capabilities in forming communication links. For example, creating links between UAVs can be more expensive than links between soldiers as the distance between UAVs can be much longer. In addition, the differences in network creativity and the ordering of the two network players can affect the outcome of the designed network. For clarity, we present an example in Fig. \ref{example_intro} to illustrate the dynamic game model considered in this work.

\begin{figure}[!t]
\begin{centering}
\includegraphics[width=1\columnwidth]{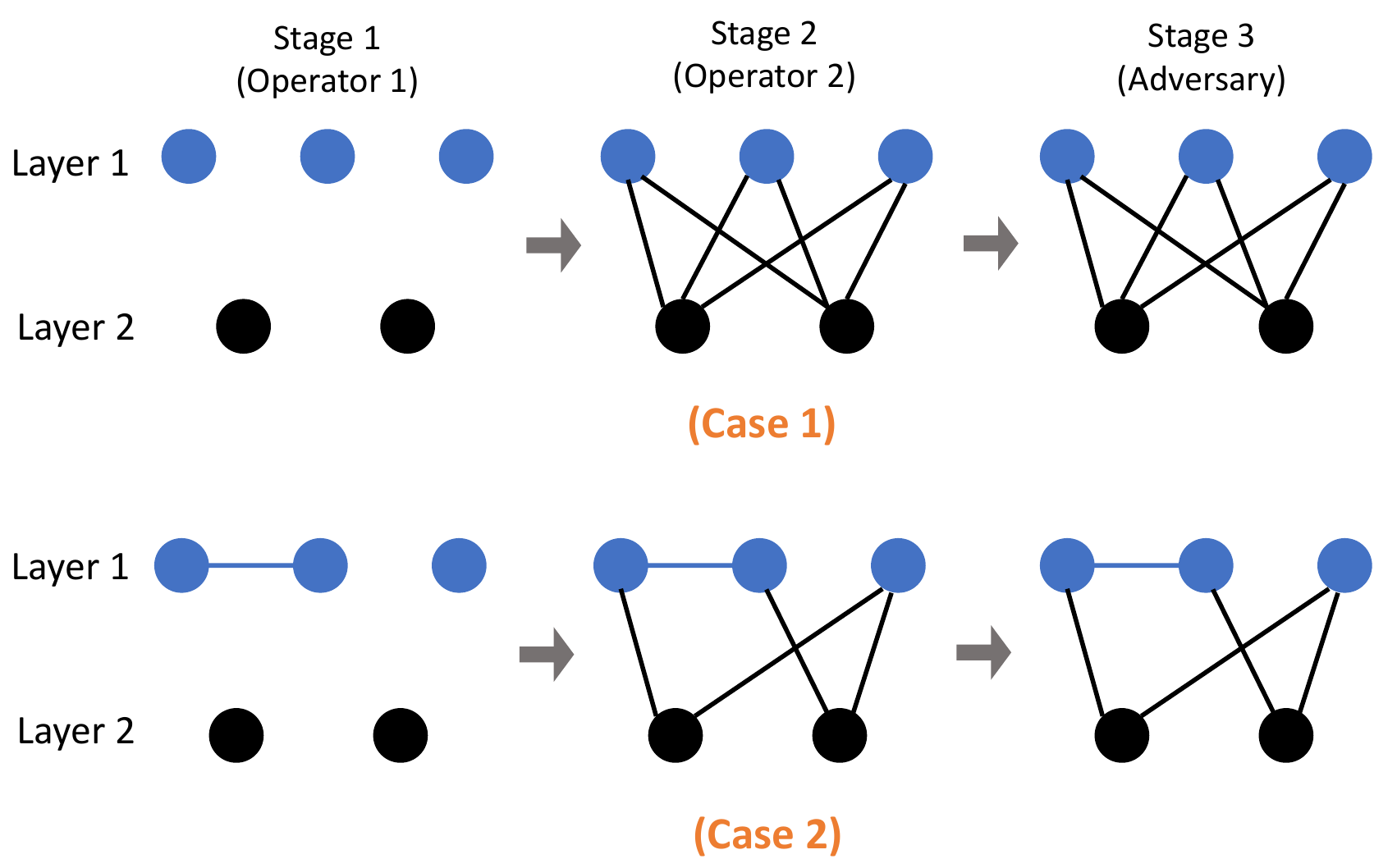}
\par\end{centering}     
\caption{\label{example_intro}
Illustration of the three-player three stage game.  Operator 1 can construct links within layer 1 and between two layers, while operator 2 can establish links within layer 2 and between two layers. Case 1: Operator 1 anticipates that operator 2 has the capability to establish at least 6 links and the attacker is able to disrupt 1 link. Thus, operator 1 does not create link at the first stage, and operator 2 creates 6 links at the second stage, where the global network is resistant to any 1 link removal attack. At stage 3, the adversary does not attack as 1 link removal cannot disconnect the network. Case 2: Operator 2 is able to construct at most 5 links. Thus, operator 1 needs to construct 1 link at the first stage by anticipating 1 link removal, and operator 2 creates 4 links at the second stage. The attacker does not compromise any link} and the network is connected at stage 3.
\end{figure}

We adopt subgame perfect Nash equilibrium (SPE) \cite[Chapter~5]{Osborne} as the solution concept to the three-player sequential IoT-enabled infrastructure network design game. We first observe that the SPE of the  game results in a $k$-connected graph  if the network remains connected at equilibrium. To understand the efficiency of the Nash equilibrium network, we use a centralized network design problem as a benchmark in which both operators coordinate and design an optimal secure network as a team.  
We further observe that the price of anarchy (PoA) can be unbounded in general cases. However, when two subnetworks contain the same number of nodes and the unitary costs of creating links are the same with only a non-null strategy of operator 2 at SPE, then the PoA is 1, which means that the decentralized network design is as efficient as that of centralized one. 
Some counter-intuitive results are further presented in Section \ref{counter_results}, e.g., the payoff of operator 1 is unique at SPE while operator 2's may vary.
 Finally, we use case studies on IoBT to illustrate the design principles of secure infrastructure networks. We observe that with a higher threat level, the two network operators prefer more collaborations to secure the IoBT network.
 
 The contributions of this paper are summarized as follows:
\begin{enumerate}
\item We propose a three-player three-stage dynamic game  to investigate the secure design of interdependent IoT-enabled infrastructure networks. By considering costs for creating and compromising links, the network operators aim to maintain the network connected while the attacker's goal is to disconnect the network.
\item We adopt subgame perfect Nash equilibrium (SPE) as the solution concept for the established game, and analyze the SPE systematically including the optimal strategies of two network operators and the attacker through backward induction.
\item We quantify the efficiency of the SPE by comparing it with the team optimal solution in which two network operators can coordinate. Furthermore, we derive a number of meaningful results including the non-unicity of equilibrium cost of network operators and the price of seniority quantifying the benefits of playing first in the dynamic game.
\item We design a scalable algorithm for constructing the secure interdependent IoT-enabled infrastructure network. In addition, we use case studies of IoBT to illustrate the derived design principles.
\end{enumerate}

\subsection{Organization of the Paper}
The rest of the paper is organized as follows. Section \ref{related_work_sec} presents related works. Section \ref{formulation} introduces the  IoT-enabled infrastructures framework and formulates the dynamic game problem. Equilibrium analysis of the game and structural results are presented in Section \ref{analysis}. Section \ref{construct_alg} designs an algorithm to guide the secure solution network construction. Case studies on IoBT networks are provided in Section \ref{case_study}, and Section \ref{conclusion} concludes this paper.

\section{Related Work}\label{related_work_sec}
With the increasing amount of cyber threats, security becomes a critical concern for IoT-enabled infrastructure networks \cite{brown2006defending,lewis2014critical,chen2019interdependent,alpcan2010network}. The infrastructure networks face various types of attacks in terms of the potential of cyber attackers \cite{Nia}. For example, attackers can target the edge computing nodes in IoT including sensor nodes. Some typical adversarial scenarios include node replication attack \cite{Parno} and DoS attack \cite{zhang2018attack}. The attackers can also launch attacks through the IoT communication networks.  In the eavesdropping attack, the attacker captures the private information over the channel, and utilizes the information to design other tailored attacks \cite{Mukherjee}. In the data injection attack, the attacker injects fraudulent packets into IoT communication links through insertion, manipulation, and replay techniques \cite{Zhou4625802}.  In our work, we focus on the adversarial attacks which lead to communication link removal in the IoT-enabled infrastructure networks.

Different methodologies have been proposed in previous works to address the cyber threats in the IoT and critical infrastructures \cite{Abomhara}. In  \cite{chen2017,zhang2017bi,khalili2018designing},
contract and insurance mechanisms have been adopted to mitigate the cyber risks with a consideration of dependencies between different entities. In
\cite{xiao2018iot}, the authors have provided a comprehensive review on the machine learning-based security schemes in the IoT systems. Another widely adopted approach for investigating infrastructure network security is through dynamic games. For example, the authors in \cite{wang2012stochastic} have proposed stochastic game nets for stochastic games representation with Petri nets to study the security analysis for enterprise networks. \cite{chen2019dynamic} has proposed a dynamic game model including pre-attack defense and post-attack recovery phases in designing resilient IoT-enabled infrastructure networks.  In \cite{nguyen2009stochastic}, the authors have adopted a two-player stochastic game to investigate the network security in which the nodes' security assets and vulnerabilities are correlated. \cite{chen2019TCNS} has investigated optimal design of two-layer IoT network with heterogeneous security considerations at different layers. In our work, we develop a three-player three-stage dynamic game framework and study the operators' and attacker's strategies in defending and compromising the network sequentially.

Infrastructure network connectivity is critical in information exchange in many civilian and military applications. Regarding the technical aspect, we investigate the secure design of IoT-enabled infrastructure network by considering the connectivity \cite{Marcin,chen2020control,Christophe} using graph theory \cite{gross2004handbook} and dynamic games \cite{Osborne,basar1999dynamic}. Different from \cite{chen2019dynamic} which has adopted a dynamic game approach in designing secure infrastructure, this work extends the single network paradigm to various network operators by considering the heterogeneity of multi-layer networks. Comparing with  previous works \cite{Marcin,Christophe} that have focused on a single-layer adversarial network design, we establish a two-layer heterogeneous network framework and characterize the decentralized decision-making of each network operator explicitly. Furthermore, the current work extends our previous one \cite{chen2017heterogeneous} in multiple aspects.  First, we include much more analytical results with investigations on the properties of the game in Section \ref{analysis}. Specifically, more results on the efficiency of equilibrium are presented in Section \ref{efficiency_subsection}, and the findings in both Sections \ref{counter_results} and \ref{collaboration_results} are completely new. Second, we provide a complete analysis of the algorithm including its complexity and scalability in Section \ref{construct_alg}. Third, we include detailed proofs of all theoretical results which were omitted in \cite{chen2017heterogeneous}.  Fourth, we extensively expand the introduction and related work sections as well as the case studies with more examples and discussions to corroborate the obtained results.

\section{Interdependent IoT-Enabled Infrastructure Model and Problem Formulation}\label{formulation}

We consider two infrastructure network operators and two sets of nodes ${\cal N}_1$ and ${\cal N}_2$, where nodes represent the devices and infrastructures in the IoT-enabled network. The first operator controls nodes in ${\cal N}_1$ and as such can create wireless communication links between those nodes as well as links connecting a node in ${\cal N}_1$ to one in ${\cal N}_2$. Similarly, the second operator controls nodes in ${\cal N}_2$ and can create links except those in ${\cal N}_1$. For convenience, we define the following notations:
\begin{itemize}
\item $E_1$ is the set of possible links between nodes of ${\cal N}_1$, that is $E_1 = \{e_1 = (n_a, n_b): n_a \in {\cal N}_1, \, n_b \in {\cal N}_1, n_a\neq n_b\}$.
\item $E_2$ is the set of possible links between nodes of ${\cal N}_2$, that is $E_2 = \{e_2 = (n_a, n_b): n_a \in {\cal N}_2, \, n_b \in {\cal N}_2,  n_a\neq n_b\}$.
\item $E_{1,2}$ is the set of possible links between nodes of ${\cal N}_1$ and ${\cal N}_2$, that is $E_{1,2} = \{e_{1,2} = (n_a, n_b): n_a \in {\cal N}_1, \, n_b \in {\cal N}_2\}$.
\end{itemize}
Further, the number of nodes in two networks are equal to $|\mathcal{N}_1|=n_1$ and $|\mathcal{N}_2|=n_2$, respectively, where $|\cdot|$ denotes the cardinality of a set. Note that the created links by both network operators are assumed to be undirected, i.e., two nodes can communicate bidirectionally if they are connected with a link.

The adversarial IoT-enabled infrastructure network formation consists of three stages which are as follows.
\begin{itemize}
\item[1)] At round $1$, operator $1$ has the choice of creating a set of communication links in $E_1 \cup E_{1,2}$.
\item[2)] At round $2$, operator $2$ can create a set of communication links in $E_2 \cup E_{1,2}$.
\item[3)] At round $3$, an adversary can remove a set of communication links, e.g., through jamming attacks, that have been created during the previous two rounds.
\end{itemize}

Note that the order of play between two network operators are determined before the game begins. As shown in Section \ref{counter_results} later, playing first or second impacts the utility of two network operators in this dynamic game.

A network is a pair $(\cal N, \cal E)$, with $\cal N$ a set of nodes and $\cal E$ a set of edges, or links between two nodes. At round $1$, starting from an empty network $({\cal N}, \emptyset)$, with ${\cal N} = {\cal {N}}_1\cup \mathcal{N}_2$, operator $1$ creates a set ${\cal E}_1:={\cal E}_1^1\cup {\cal E}_1^{1,2}$ of links and thus designs network $G_1 = ({\cal N}, {\cal E}_1)$ such that ${\cal E}_1$ is a subset of $E_1 \cup E_{1,2}$, the set of admissible links for operator $1$, i.e., ${\cal E}_1^1 \subseteq E_1$ and ${\cal E}_1^{1,2}\subseteq E_{1,2}$.
Then, at round $2$, starting from network $G_1$, operator $2$ creates a set ${\cal E}_2:={\cal E}_2^2\cup {\cal E}_2^{2,1}$ of links and thus designs network $G_2 = ({\cal N}, {\cal E}_1 \cup {\cal E}_2)$ such that ${\cal E}_2^2 \subseteq E_2$ and ${\cal E}_2^{2,1} \subseteq E_{1,2}$.
Finally, at round $3$, the adversary chooses a subset of the links ${\cal E}_A \subseteq {\cal E}_1 \cup {\cal E}_2$ that it removes from $G_2$, resulting in network $G_3 = ({\cal N}, {\cal E}_1 \cup {\cal E}_2 \backslash {\cal E}_A)$. We denote the total number of nodes as $|{\cal N}|=n_1+n_2 = n$.

 One metric to measure the performance of the IoT-enabled infrastructure network is its connectivity defined as follows.

\begin{definition}[Network Connectivity]
An infrastructure network $G=(\mathcal{N},\mathcal{E})$ is said to be connected if every node in $\mathcal{N}$ can be reached from any others through a sequence of links from $\mathcal{E}$.
\end{definition}

The goal of the operators is to construct a connected IoT-enabled infrastructure network, i.e., a network where every node can be reached from any others through a sequence of links. Conversely, the role of the adversary is to obtain a disconnected infrastructure network, and thus a node or a group of nodes becomes not accessible to the rest of the network. Note that the adopted binary network connectivity metric is suitable for mission-critical IoT-enabled infrastructures, e.g., IoBT. In these types of systems, the connectivity among agents is the minimum requirement and thus of first priority of the network operator.

In IoT-enabled infrastructure network, both creating and removing links are costly. Let $c_{1}$ and $c_{2}$ be the normalized unitary costs for creating a link for operators $1$ and $2$ in $E_1$ and $E_2$, respectively, and $c_{1,2}$ and $c_{2,1}$ be their corresponding normalized unitary costs for creating a link in $E_{1,2}$. In addition, $c_{A}$ is the normalized unitary cost of the adversary to compromise a link.  Let $\mathds{1}_G$ be the indicator factor that equals $1$ if network $G$ is connected and $0$ otherwise. Then, the payoffs of operators $1$, $2$ and the adversary are, respectively,
\begin{equation}\label{utility}
\begin{split}
U_{1}({\cal E}_1, {\cal E}_2, {\cal E}_A) &= \mathds{1}_{({\cal N}, {\cal E}_1 \cup {\cal E}_2 \backslash {\cal E}_A)} - c_{1} |{\cal E}_1^1| - c_{1,2} |{\cal E}_1^{1,2} |,\\
U_{2}({\cal E}_1, {\cal E}_2, {\cal E}_A) &=\mathds{1}_{({\cal N}, {\cal E}_1 \cup {\cal E}_2 \backslash {\cal E}_A)} - c_{2} |{\cal E}_2^2| - c_{2,1} |{\cal E}_2^{2,1}|,\\
U_{A}({\cal E}_1, {\cal E}_2, {\cal E}_A) &= 1-\mathds{1}_{({\cal N}, {\cal E}_1 \cup {\cal E}_2 \backslash {\cal E}_A)} - c_{A} |{\cal E}_A|.
\end{split}
\end{equation}
In \eqref{utility}, we normalize the cost of creating and attacking links with respect to the unitary network connectivity measure. Note that these constant cost parameters are directly related to the monetary cost of creating and attacking links. To avoid trivial solutions, the unitary costs $c_1$, $c_2$, $c_{1,2}$, $c_{2,1}$, and $c_A$ admit a value between 0 and 1. Both operators adopt the same type of communication technology (e.g., WiFi, Zigbee, LTE, NB-IoT) to construct the links. Thus, both intralinks and interlinks in the two-layer network have an identical level of security. This fact leads to a unified unitary attack cost for the adversary, regardless of the type of communication links (i.e., intralink or interlink). Extension of the framework to include heterogeneous attack costs is also possible. One direction is to consider that the operators can construct secure (with protection mechanism) and insecure links. The attacker is incapable of compromising those secure links or very costly to do so. We leave this nontrivial extension as future work.

In this work, we are interested in seeking the subgame perfect Nash equilibrium (SPE) of the three-stage dynamic game. SPE is a refinement of  Nash equilibirum (NE) by focusing on the sequential strategic decision-making of players, comparing with NE that treats all players' strategies as choices made once before the play begins. An SPE has the property that in \textit{no subgame} can any player benefit by unilaterally deviating from the SPE strategy.  Another critical property is that every SPE is an NE (but not vise versa), as SPE is a more restricted solution concept. Furthermore, SPE is a strategy profile that induces an NE in every subgame. A more comprehensive discussion of SPE can be found in \cite[Chapter~5]{Osborne} and \cite[Chapter~5]{fujiwara2015non}. In this regard, SPE is an appropriate solution concept for our three-stage game. Specifically, we seek for optimal strategies of the three players as follows.
\begin{definition}[SPE]
The SPE of the investigated three-player three-stage game is a strategy profile $\{{\cal E}_1^{\text{MAX}}, {\cal E}_2^{\text{MAX}},{\cal E}_A^{\text{MAX}}\}$ that satisfy the following constraints.
\begin{itemize}
\item[1)] Given network $G_2$, the adversary chooses the optimal set ${\cal E}_A^{\text{MAX}}({\cal E}_1, {\cal E}_2) $ that maximizes its utility ${\cal E}_A^{\text{MAX}}({\cal E}_1, {\cal E}_2) \in \argmax_{{\cal F}} \{U_{A}({\cal E}_1, {\cal E}_2, {\cal F})\}$.
\item[2)] Given network $G_1$, operator $2$ chooses the optimal set ${\cal E}_2 ^{\text{MAX}}({\cal E}_1)$ that maximizes its utility ${\cal E}_2^{\text{MAX}}({\cal E}_1) \in \argmax_{{\cal F}} \{U_{2}({\cal E}_1, {\cal F}, {\cal E}_A^{\text{MAX}}({\cal E}_1, {\cal F}))\}$.
\item[3)] Operator $1$ chooses the optimal set ${\cal E}_1^{\text{MAX}}$ that maximizes its utility ${\cal E}_1^{\text{MAX}} \in \argmax_{{\cal F}} \{U_{1}({\cal F}, {\cal E}_2^{\text{MAX}} ({\cal F}), {\cal E}_A^{\text{MAX}}({\cal F}, {\cal E}_2^{\text{MAX}} ({\cal F})))\}$.
\end{itemize}
\end{definition}

By convention, the adversary attacks the network when $c_{A} |{\cal E}_A|=1$ and $\mathds{1}_{({\cal N}, {\cal E}_1 \cup {\cal E}_2 \backslash {\cal E}_A)}=0$ at SPE. This consideration makes sense, especially in mission-critical scenarios, as the attacker's goal is to disrupt the network connectivity. Attacking the network with success can still be an incentive-compatible strategy for the adversary, even if it yields a zero net utility. However, this convention is not strict, i.e., the adversary may not attack the network in this boundary case. To keep the notation simple during the SPE analysis (e.g., characterization of $k$ in \eqref{notation-k}), we follow such convention. Note that the developed theoretical methods in later sections are valid under both considerations. In contrast, the operators will not secure the network if $c_{1} |{\cal E}_1^1| + c_{1,2} |{\cal E}_1^{1,2}|=1$, $c_{2} |{\cal E}_2^2| + c_{2,1} |{\cal E}_2^{2,1}|=1$, and $\mathds{1}_{({\cal N}, {\cal E}_1 \cup {\cal E}_2 \backslash {\cal E}_A)}=1$. Therefore, the SPE yields the equilibrium topology of the two-layer adversarial IoT-enabled infrastructure networks. 

\textit{Remark:} We next comment on the information structure of game. In this dynamic game, the adversary is the follower and his optimal actions are best responses to the network constructed by the two operators after the first two steps. Thus, the adversary does not necessarily know the exact link creation costs of two network operators but reacts to the operators' strategy profile ${\cal E}_1\cup {\cal E}_2$ optimally by maximizing $U_A$. In comparison, both operators need to know the threat level of the network captured by the unitary attack cost $c_A$. With a smaller $c_A$, the operators strategies are more conservative by anticipating more attacks. When $c_A$ is not precisely known by operators, there are two possible ways to address this challenge. The first one is that, since $c_A$ is directly related to the number of attacks, then the operators can estimate $c_A$ based on the distribution of number of attacks using historical data. The second approach is to model this unknown information directly using a parameter $\theta$, representing the uncertain type of the adversary. Then, the operators determine their optimal strategies conditioned on the random type parameter $\theta$. This yields a new dynamic Bayesian game which is nontrivial and worth of investigation in the future work.

\section{SPE Analysis and Analytical Results}\label{analysis}
In this section, we analyze the formulated three-player three-stage IoT-enabled infrastructure network formation game  in Section \ref{formulation} with a focus on the derivation of its SPE. 
\subsection{Backward Induction}\label{backward_induction}

To derive the SPE, we proceed by backward induction, i.e., we first compute the optimal strategy for the adversary, then operator $2$ and finally operator $1$.

\begin{definition}[$p$-Connected Network]
 A network $G=(\cal N,\cal E)$ is called a $p$-connected network, if (i) it remains connected after the deletion of any $p$ links, and (ii) there exists a set $\cal F$ of $p+1$ links ($|\mathcal F| = p+1$) so that the network $(\mathcal{N},\mathcal{E}\backslash \cal{F})$ is disconnected.
\end{definition}

\textit{Remark:} Any connected network $G=(\cal N,\cal E)$ is a $p$-connected network for some value of $1 \leq p \leq |\cal E|$. By convention, we say that a non-connected network is $(-1)$-connected. The value $p$ is called the \textit{link connectivity} of the network. In addition, a network is \textit{$p$-resistant} if it remains connected after the deletion of $p$ links, that is, if it is $m$-connected for some $m \geq p$. 

For notational clarity, in the following, we denote
\begin{equation}\label{notation-k}
k := \floor*{1/c_A}.
\end{equation}
Note that the floor operator $\floor*{x}$ and the ceiling operator $\ceil*{x}$ denote the largest integer no greater than $x$ and the smallest integer no less than $x$, respectively. As shown in the analysis of following Lemma \ref{lem:1}, $k$ can be interpreted as the maximum number of links that the adversary may attack at the SPE. The following result presents the strategy of the adversary.

\begin{lemma}\label{lem:1}
Let ${\cal E}_1$ and ${\cal E}_2$ be played by operator $1$ and $2$ respectively.  
Then, the adversary's optimal strategy ${\cal E}_A^{\text{MAX}}$ is:
\begin{itemize}
\item[1)] $\emptyset$ if $({\cal N}, {\cal E}_1 \cup {\cal E}_2)$ is not connected,
\item[2)] $\emptyset$ if $({\cal N}, {\cal E}_1 \cup {\cal E}_2)$ is $m$-connected with $m \geq k$,
\item[3)] any ${\cal F}$ such that $\mathds{1}_{({\cal N}, {\cal E}_1 \cup {\cal E}_2\backslash \cal F)}=0$ and $\displaystyle |{\cal F}| = m+1$ if $({\cal N}, {\cal E}_1 \cup {\cal E}_2)$ is $m$-connected with $m<k$.
\end{itemize} 
\end{lemma}
\begin{proof}
Note that since $1-\mathds{1}_{({\cal N}, {\cal E}_1 \cup {\cal E}_2)} \in \{0,1\}$ then the utility of the adversary is upper bounded by $1$. 
Further, if $\mathds{1}_{({\cal N}, {\cal E}_1 \cup {\cal E}_2)} = 0$, i.e., if $({\cal N}, {\cal E}_1 \cup {\cal E}_2)$ is not connected, then $U_{A}({\cal E}_1, {\cal E}_2, {\cal E}_A) = 1$ if and only if (iff) $|{\cal E}_A| = 0$, which is thus the (only) optimal strategy.

Assume now that $\mathds{1}_{({\cal N}, {\cal E}_1 \cup {\cal E}_2)} = 1$. Note that $U_{A}({\cal E}_1, {\cal E}_2, \emptyset) = 1-\mathds{1}_{({\cal N}, {\cal E}_1 \cup {\cal E}_2)}=0$. Thus, the optimal strategy of the adversary is not the empty set iff there exists an ${\cal F} \neq \emptyset$ such that $ \mathds{1}_{({\cal N}, {\cal E}_1 \cup {\cal E}_2 \backslash {\cal F})} = 0$ and $U_{A}({\cal E}_1, {\cal E}_2, {\cal F}) >0$. 

Let $m$ be the connectivity of network $ ({\cal N}, {\cal E}_1 \cup {\cal E}_2)$. Let ${\cal F}$ be a non-empty set such that $({\cal N}, {\cal E}_1 \cup {\cal E}_2\backslash \cal F)$ is connected. Then, $U_{A}({\cal E}_1, {\cal E}_2, {\cal F}) = - c_A|{\cal F}| <0$. 
Thus this strategy is strictly dominated by the null strategy and therefore is not optimal. Reciprocally, let ${\cal F}$ be such that $({\cal N}, {\cal E}_1 \cup {\cal E}_2\backslash \cal F)$ is disconnected. Then $U_{A}({\cal E}_1, {\cal E}_2, {\cal F}) = 1- c_A|{\cal F}| \leq 1-(m+1) c_A$. 
Thus, the null strategy is optimal iff $k \leq 1/c_A < m+1$, and a non-null strategy is optimal iff $m+1 \leq 1/c_A < k+1$, i.e., if $m<k$.
\end{proof}

In a similar vein, we can detail the optimal strategy of operator $2$ as follows.

\begin{lemma}\label{lem:2}
Let ${\cal E}_1$ be played by operator $1$. 
Then, the operator 2's optimal strategy ${\cal E}_2^{\text{MAX}}$ in the IoT-enabled infrastructure is:
\begin{itemize}
\item[1)] $\emptyset$ if $({\cal N}, {\cal E}_1)$ is $k$-connected;
\item[2)] Otherwise, let $F$ be the set of sets of $E_2 \cup E_{1,2}$ such that for each element ${\cal F}$ of $F$, network $({\cal N}, {\cal E}_1 \cup {\cal F})$ is $k$-connected. If $F$ is not empty, we consider its element ${\cal F}$ that has the minimal cost, that is the set of links $\displaystyle {\cal F} = {\cal F}_2\cup {\cal F}_{2,1}$, where $\{ {\cal F}_2, {\cal F}_{2,1}\}= \argmin_{{\cal A}_2, {\cal A}_{2,1}} \{ c_2 |{\cal A}_2| + c_{2,1} |{\cal A}_{2,1}|,\ \text{s.t. } {\cal A}_2 \subseteq E_2, \, {\cal A}_{2,1} \subseteq E_{1,2} \}$.
\begin{itemize}
\item If $F = \emptyset$ or $c_2 |{\cal F}_2| + c_{2,1} |{\cal F}_{2,1}| \geq 1$, then the optimal strategy of operator 2 is the null strategy and the resulting payoff is $0$.
\item Otherwise (i.e. $F$ is not empty and $c_2 |{\cal F}_2| + c_{2,1} |{\cal F}_{2,1}| < 1$), then the optimal payoff of operator $2$ is $1-c_2 |{\cal F}_2| - c_{2,1} |{\cal F}_{2,1}|$, and an optimal strategy is ${\cal F}$.
\end{itemize}
\end{itemize} 
\end{lemma}

This result leads us finally to the optimal strategy for operator $1$ in the IoT-enabled infrastructure networks.

\begin{lemma}\label{lem:3}
 Let ${\cal G}$ be the set of $k$-connected networks. (Note that ${\cal G}$ is not empty iff $k+1 \leq n-1$.) Any network of ${\cal G}$ can be written in the form $({\cal N}, {\cal E}_1^1 \cup {\cal E}^{1,2}_{1} \cup {\cal E}_2^2  \cup {\cal E}^{2,1}_{2})$ with ${\cal E}_1^1 \subseteq E_1$, ${\cal E}_2^2 \subseteq E_2$, ${\cal E}^{1,2}_{1} \subseteq E_{1,2}$ and ${\cal E}^{2,1}_{2} \subseteq E_{1,2}$. Now, let ${\cal \tilde{G}} \subseteq {\cal G}$ be the subset of $k$-connected networks that lead to positive utilities for operator $1$ and $2$, that is networks ${\cal \tilde{G}}=({\mathcal N},{\cal \tilde{E}}_1^1\cup {\cal \tilde{E}}_2^2\cup  {\cal \tilde{E}}^{1,2}_{1}\cup {\cal \tilde{E}}^{2,1}_{2})$ such that $c_1 |{\cal \tilde{E}}_1^1| + c_{1,2}{\cal |\tilde{E}}^{1,2}_{1}| < 1$ and $c_2 |{\cal \tilde{E}}_2^2| + c_{2,1} |{\cal \tilde{E}}^{2,1}_{2}| < 1$. Then, the optimal strategy of the first operator ${\cal E}_1^{\text{MAX}}$ is:
\begin{itemize}
\item[1)] $\emptyset$ if ${\cal \tilde{G}} = \emptyset$, and the associated payoff is $0$.
\item[2)] the elements of ${\cal \tilde{G}}$ that have the minimal value of $c_1 |{\cal \tilde{E}}_1^1| + c_{1,2}{\cal |\tilde{E}}^{1,2}_{1}|$ otherwise.
\end{itemize} 
\end{lemma}

From Lemmas~\ref{lem:1}, \ref{lem:2}, and \ref{lem:3}, we can finally deduce the SPE.
\begin{lemma}\label{lem:4}
 Let ${\cal G}$ be the set of $k$-connected networks. Now, let ${\cal \tilde{G}} \subseteq {\cal G}$ be the subset of $k$-connected networks that lead to positive utilities for operator $1$ and $2$, that is networks $\cal \tilde{G}$ such that $c_1 |{\cal \tilde{E}}_1^1| + c_{1,2}{\cal |\tilde{E}}^{1,2}_{1}| < 1$ and $c_2 |{\cal \tilde{E}}_2^2| + c_{2,1} |{\cal \tilde{E}}^{2,1}_{2}| < 1$. We obtain the following results.
\begin{itemize}
\item[1)] If ${\cal \tilde{G}} = \emptyset$, then the optimal strategy for operator $1$ and $2$ and adversary are empty sets and the resulting utilities are $U_1(\emptyset, \emptyset, \emptyset) = U_2(\emptyset, \emptyset, \emptyset) = 0$ and $U_A(\emptyset, \emptyset, \emptyset) = 1$.
\item[2)] Otherwise, the optimal strategies of operator $1$ are the elements of ${\cal \tilde{G}}$ that have minimal value of $c_1 |{\cal \tilde{E}}_1^1| + c_{1,2}{\cal |\tilde{E}}^{1,2}_1|$. Then, if ${\cal \tilde{E}}_1^1 \cup {\cal \tilde{E}}^{1,2}_{1}$ is the strategy of operator $1$, the optimal strategies of operator $2$ are the elements of ${\cal \tilde{G}}$ of the form $(\mathcal{N}, {\cal \tilde{E}}_1^1  \cup {\cal \tilde{E}}^{1,2}_{1} \cup {\cal \tilde E}_2^2 \cup {\cal \tilde E}^{2,1}_{2})$ with ${\cal \tilde E}_2^2 \subseteq E_2$ and ${\cal \tilde E}^{2,1}_{2} \subseteq E_{1,2}$ that minimizes $c_2 |{\cal \tilde{E}}_2^2| + c_{2,1} |{\cal \tilde{E}}^{2,1}_{2}|$. Finally, the optimal strategy for the adversary is the null strategy leading to $U_A({\cal \tilde{E}}_1^1 \cup {\cal \tilde E}^{1,2}_{1}, {\cal \tilde{E}}_2^2 \cup {\cal \tilde E}^{2,1}_{2}, \emptyset) = 0$. 
\end{itemize}
\end{lemma}

In the following, we denote the SPE in the following format: $w := (({\cal E}_1^1,{\cal E}^{1,2}_1),({\cal E}_2^2,{\cal E}^{2,1}_2), {\cal E}_A)$ with ${\cal E}_1^1 \cup {\cal E}^{1,2}_1$ the strategy of operator 1, ${\cal E}_2^2 \cup {\cal E}^{2,1}_2$ the strategy of operator $2$, and ${\cal E}_A$ the strategy of the adversary (with ${\cal E}_A \subseteq \mathcal E_{1}^1 \cup {\cal E}^{1,2}_1 \cup {\cal E}_2^2 \cup {\cal E}_2^{2,1} $).

We can thus draw the following result.
\begin{lemma}\label{lem:5}
From Lemma~\ref{lem:4}, we obtain that
the only SPE is the null strategy for the three players, that is the SPE is $((\emptyset, \emptyset), (\emptyset, \emptyset), \emptyset)$ if any of the following condition is satisfied:
\begin{itemize}
\item[1)] $n_1 +n_2-1 < k+1$;
\item[2)] $(k+1)\min(c_{1,2}, c_{2,1}) + \floor*{\frac{(n_1-1)(k+1)}{2}} c_1 +  \floor*{\frac{(n_2-1)(k+1)}{2}} c_2\geq 2$ under $c_1 \leq c_{1,2}$ and $c_2 \leq c_{2,1}$.
\end{itemize}
In these cases, the SPE also corresponds to the optimal strategy for each of the $3$ players.
\end{lemma}

\begin{proof}
We prove successively the two conditions:\\
(i)  $n_1 +n_2-1$ represents the maximal number of nodes that each node can connect to, i.e., the maximal degree. If this value is lower than $k+1$, then the adversary is able to disconnect any given network with at most $k$ link removals.\\
(ii) Note that for the network to be $k$-resistant, each node should have a degree of (at least) $k+1$. Since there are $n_1+n_2$ nodes, then there are at least $\ceil*{\frac{(n_1+n_2)(k+1)}{2}}$ links in the network. Further, ${\cal E}^{1,2}_1 \cup {\cal E}^{2,1}_2$ should contain at least $k+1$ links. Since $c_1 \leq c_{1,2}$ and $c_2 \leq c_{2,1}$, then these links are the ones with maximal cost. Similarly, there are at least $\floor*{\frac{(n_1-1)(k+1)}{2}}$ links in ${\cal E}_1^1 \cup {\cal E}^{1,2}_1$ and at least $\floor*{\frac{(n_2-1)(k+1)}{2}}$ links in ${\cal E}_2^2 \cup {\cal E}^{2,1}_2$.
\end{proof}

In the following, we thus focus our attention in situations in which the conditions of Lemma~\ref{lem:5} are not satisfied. Furthermore, we denote the set of SPE of the game by $\mathcal{L}^*$.

\subsection{Efficiency of the Equilibria}\label{efficiency_subsection}

In this section, we are interested in how different the costs are at the SPE and in a system where both operators can coordinate. In the scenario where both players can coordinate, the problem amounts to finding the optimal solution, and it uses at least $\ceil*{\frac{(n_1+n_2)(k+1)}{2}}$ links. Note that this bound can be reached using Harary networks \cite{harary}. However, since the costs $c_1$, $c_2$ and $c_{1,2}$ are different, the Harary networks using the least number of links may not correspond to the ones with the lowest cost. We present the upper and lower bounds of costs in the following proposition.

 \begin{proposition}[Upper and Lower Bounds]
 A lower bound on the total cost for creating a network is
 $$
 (k+1) \min(c_{1,2},c_{2,1}) + \ceil*{\frac{(n_1+n_2-2)(k+1)}{2}} \min (c_1, c_2).
 $$

 Suppose that $n_1 \geq k+1$ and $n_2 \geq k+1$. An upper bound on the total minimal cost for creating a network is
 $$
 (k+1)\min(c_{1,2}+c_{2,1}) + \ceil*{\frac{n_1(k+1)}{2}}c_{1} + \ceil*{\frac{n_2(k+1)}{2}}c_{2}.
 $$
 \end{proposition}

 \begin{proof} We prove successively the two parts.

 (Lower Bound): We know that since any node needs to have (at least) a degree of $k+1$, then at least $\ceil*{\frac{(n_1+n_2)(k+1)}{2}}$ links need to be created, among which there should be at least $k+1$ in ${\cal E}_{1,2} \cup {\cal E}_{2,1}$ so that the adversary cannot disconnect nodes of ${\cal N}_1$ from nodes of ${\cal N}_2$.
 
  (Upper bound): Since $n_1 \geq k+1$ and $n_2 \geq k+1$, then construct a $(n_1, k+1)$-Harary network among nodes of ${\cal N}_1$ and a $(n_2, k+1)$-Harary network among nodes of ${\cal N}_2$ which require $\ceil*{\frac{n_1(k+1)}{2}}$ and $\ceil*{\frac{n_2(k+1)}{2}}$ links, respectively. Finally, construct $k+1$ links between distinct nodes of  ${\cal N}_1$ and nodes of ${\cal N}_2$ for the global network being secure.
 \end{proof}

From Lemma~\ref{lem:4}, at the SPE, the two operators sequentially form an IoT-enabled infrastructure network that is $k$-connected (if such network can be constructed so that they both receive a positive utility). Recall that $w = (({\cal E}_1^1,{\cal E}^{1,2}_1),({\cal E}_2^2,{\cal E}^{2,1}_2), {\cal E}_A)$ with ${\cal E}_1^1 \cup {\cal E}^{1,2}_1$, ${\cal E}_2^2 \cup {\cal E}^{2,1}_2$, and ${\cal E}_A$ denoting the strategies of the first operator, second operator, and  the adversary, respectively. Then, the definition of price of anarchy (PoA) is as follows.

\begin{definition}[Price of Anarchy]
The PoA for the secure IoT-enabled  infrastructure network formation game is defined as
\begin{align}
PoA = \max_{w\in\mathcal{L}^*} \frac{C_{SPE}(w)}{C_{CO}},
\end{align}
where $C_{SPE}$ and $C_{CO}$ are the sum of costs for the operators at the SPE network and the sum of costs they would experience with coordination, respectively. 
\end{definition}

The following proposition shows that the individual costs  as well as the global sum of costs can be arbitrarily different in the SPE and coordinated optimal infrastructure networks.

\begin{proposition}
The PoA of the secure IoT-enabled infrastructure network formation game can be unbounded.
\end{proposition}

\begin{proof}
We show the result by considering a situation with $n_1 >2$, $n_2 = 2$, $k=1$, $c_1=c_2=\frac{1}{n_1^3}$, $c_{1,2}=\frac{1}{n_1^2}$ and $c_{2,1}=\frac{1}{3n_1}$.

An optimal joint strategy is to create all links of the form $(i,i+1)$ with $1 \leq i \leq n_1+2$ and link $(n_1+2,1)$. As this construction forms a cycle of the $n_1+n_2$ nodes, then it is $1$-connected. Further, it contains exactly $\ceil*{\frac{(n_1+n_2)(k+1)}{2}}=n_1+n_2$ links, and among those $k+1=2$ are in $E_{1,2}$. It is therefore an optimal solution. Its cost is $C_{CO} = (n_1+n_2-2)c_1+2c_{1,2} = n_1 c_1+2c_{1,2} = \frac{3}{n_1^2}$.

Next, we investigate the SPE of the game. The operator $1$ plays the null strategy only if operator $2$ can construct a $1$-connected network at a cost lower than $1$. We then consider the following strategy for operator $2$ that consists in creating all links of the form $(i, n_1+1)$ and $(i, n_1+2)$ for all $1 \leq i \leq n_1$. This strategy has a cost of $C_{SPE} = 2 n_1 c_{2,1} = \frac{2}{3} < 1$, and the resulting network is $1$-connected which can be shown by using Menger's theorem \cite{Menger}. Indeed, for any nodes $i$ and $j$, we can construct at least two disjoint paths. For instance, if $i$ and $j$ are both in ${\cal N}_1$, we consider the paths (both of length $2$) $i;(i, n_1+1);n_1+1 ; (n_1+1, j) ; j$ and $i ; (i, n_1+2) ; n_1+2 ; (n_1+2, j) ; j$. If $i \in {\cal N}_1$ and $j \in {\cal N}_2$, we consider the paths of length $1$: $i ; (i, j) ; j$ and the paths of length $3$: $i; (i, k) ; k; (k, \ell) ; \ell ; (\ell, j); j$ (with $k \in {\cal N}_1$, $k \neq i$ and $\ell \in {\cal N}_2$, $\ell \neq j$). Finally, if both $i$ and $j$ are in ${\cal N}_2$, then we consider the paths of length $2$: $i;(i, 1);1 ; (1, j) ; j$ and $i;(i, 2);2 ; (2, j) ; j$. Note that this strategy is optimal for the second operator, since it creates $2 n_1$ links in $E_{1,2}$.

In summary, the degradation of performance in terms of PoA in this scenario is $C_{SPE} / C_{CO} = \frac{2n_1^2}{9}$ which increases quadratically in $n_1$. 
\end{proof}

We next characterize a class of scenarios in which $PoA=1$, leading to an efficient decentralized design of secure infrastructure network.
\begin{lemma}
When two subnetworks contain the same number of nodes,  and the unitary costs satisfy the following condition $c_1=c_2 = c_{1,2}=c_{2,1}<\frac{1}{n_1(k+1)}$, and $k< n_1$, then the PoA is 1.
\end{lemma}

\begin{proof}
In the scenarios that $c_{2,1}<\frac{1}{2n_1}$, we know that at the SPE, only operator 2 creates interlinks which are solely in $E_{1,2}$. Thus, the cost at SPE is $C_{SPE} = n_1(k+1)c_{2,1}$. For the cooperative case, the network configuration is a Harary network due to the same link costs. Therefore, the total cost is $C_{CO} = \ceil*{\frac{(n_1+n_2)(k+1)}{2}}c_1$ in the optimal network. Due to $n_1=n_2$ and $c_{2,1} = c_1$, we obtain $C_{CO} = n_1(k+1)c_{2,1}c_{2,1} = C_{SPE}$, and thus $PoA = 1$ in this scenario.
\end{proof}

\subsection{Some Counter-Intuitive Results}\label{counter_results}
In this section, we present some counter-intuitive results of the IoT-enabled infrastructure network formation game.

The following Proposition \ref{prop:unicity} shows that for given system parameters, the SPE may not be unique. This, in terms of infrastructure network architecture is not surprising, since several topologies can lead to a $k$-connected network with the minimal cost. More surprisingly, however is the fact that the SPE may not be unique in terms of the costs.

\begin{proposition}[Non-Unicity of Equilibrium Cost]\label{prop:unicity}
For given values of the parameters $n_1$, $n_2$, $c_1$, $c_2$, $c_{1,2}$, $c_{2,1}$ and $c_A$, the SPE may not be unique. More precisely, at the SPE, there is a unique payoff value associated to the operator $1$, but there may be several payoff values of the operator $2$.
\end{proposition}

\begin{proof}
We show the result by providing an example with the property. Let parameters be $n_1 = 4$, $n_2 = 3$, $k=2$, and $c_1 = c_2 = c_{1,2} = c_{2,1} = 0.09$.
From the values of $c_2$ and $c_{2,1}$, operator $1$ knows that at the SPE, operator $2$ builds at most $11$ links. Since each node of ${\cal N}_1$ needs to have a degree of $3$, if operator $1$ builds no link, then operator $2$ needs to build at least $12$ links which is more than that it can bear. Thus operator $1$ needs to build at least $1$ link. Then, depending on the choice of the link created by operator $1$, operator $2$ needs to build either $10$ or $11$ links, as illustrated in Fig.~\ref{fig:nonUnic}.
\tikzset{protege/.style = {line width=4pt}}
\tikzset{attaque/.style = {style=loosely dashed}}
\tikzset{both/.style = {line width=3pt, style=loosely dashed}}
\tikzset{multi/.style = {
shape = rectangle,
line width=1pt,
inner sep = 1.5pt,
outer sep = 0pt,
minimum size = 15pt,
draw
}}
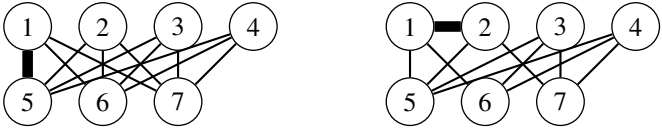
\begin{figure}
\subfloat[${\cal E}_1^1 = \emptyset$ and ${\cal E}_1^{1,2}=\{(1,5)\}$. Then ${|\cal E}_2^{2,1}|=11$.]{
\begin{tikzpicture}
\GraphInit[vstyle=Normal]
\Vertex{1} \EA(1){2} \EA(2){3} \EA(3){4}
\SO(1){5} \EA(5){6} \EA(6){7}
\Edges(5,2,6,3,7,4,6,1,7,2)
\Edges(4,5,3)
\tikzset{EdgeStyle/.style = protege}
\Edge(1)(5)
\end{tikzpicture}} \hfill
\subfloat[${\cal E}_1^1 = \{(1,2)\}$ and ${\cal E}_1^{1,2}=\emptyset$. Then ${|\cal E}_2^{2,1}|=10$.]{
\begin{tikzpicture}
\GraphInit[vstyle=Normal]
\Vertex{1} \EA(1){2} \EA(2){3} \EA(3){4}
\SO(1){5} \EA(5){6} \EA(6){7}
\Edges(6,1,5,2,7,4,6,3,5,4)
\Edge(3)(7)
\tikzset{EdgeStyle/.style = protege}
\Edge(1)(2)
\end{tikzpicture}}
\caption{$\mathcal{N}_1=\{1,2,3,4\}$ and $\mathcal{N}_2=\{5,6,7\}$. Non-unicity of the cost of operator $2$. Operator 2 either creates 10 or 11 links depending on operator 1's strategy at equilibrium. (Bold line: operator $1$; Plain thin line: operator $2$).}\label{fig:nonUnic}
\end{figure}
\end{proof}

In addition, the order of the operators creating their communication links has an impact on the payoffs of the players.
\begin{proposition}\label{prop:order}
Consider an IoT-enabled infrastructure network where the roles of the operators are symmetric, i.e.,
$ n_1 = n_2, \; c_{1,2}=c_{2,1}, \; c_1=c_2.
$
Suppose that $n_1 (k+1) c_{1,2} < 1$. Then at the SPE, the payoff of operator $1$ is $1$ while that of operator $2$ is $1-n_1(k+1) c_{1,2}$.
\end{proposition}

\begin{proof}
We consider the network $(\mathcal N, {\cal E})$ with ${\cal E} \subseteq E_{1,2}$. ${\cal E}$ is the set of links $(i,j)$ with $1 \leq i \leq n_1$ and $n_1 +(i -\floor*{\frac{k}{2}}) \mod n_2 \leq j \leq n_1 + (i+\ceil*{\frac{k}{2}}) \mod n_2$. Then, the network is $k$-connected and has exactly $\frac{(n_1+n_2)(k+1)}{2}$ links.
The operator 2 can construct a $k$-connected SPE network which has exactly $\frac{(n_1+n_2)(k+1)}{2}=n_1(k+1)$ links all in $E_{1,2}$.
\end{proof}

Next, we investigate the impact of order of play between two operators in constructing the network on their corresponding payoffs. Specifically, we propose a notion, \textit{Price of Seniority (PoS)}, to quantify the potential loss of operator 2 for playing second comparing with the scenario in which he plays first. The formal definition of PoS is presented as follows.

\begin{definition}[Price of Seniority]
Denote $C_2$ by the total cost of operator 2 at SPE, and $C_2^{(1)}$ by its total cost if it was playing first. Then the price of seniority (PoS) that quantifies the loss of operator 2 for playing second can be defined as
\begin{align}
PoS = \max_{w\in\mathcal{L}^*,\tilde{w}\in\tilde{\mathcal{L}}^*}\ \frac{C_2(w)}{C_2^{(1)}(\tilde{w})}.
\end{align}
Recall that $w$ is a strategy profile including all three players' actions, and $\mathcal{L}^*$ is a set containing SPE of the game. Similar for $\tilde{w}$ and 
$\tilde{\mathcal{L}}^*$ where the playing sequence of operators 1 and 2 is interchanged.

Furthermore, in the case of symmetric scenarios, i.e., $n_1=n_2$, $c_1=c_2$, and $c_{1,2}=c_{2,1}$, and denote $C_1$ by the total cost of operator 1, then PoS can be written as
\begin{align*}
PoS = \max_{w\in\mathcal{L}^*}\ \frac{C_2(w)}{C_1(w)} =  \max_{w\in\mathcal{L}^*}\ \frac{c_{2} |{\cal E}_2 \cap E_2| + c_{2,1} |{\cal E}_2 \cap E_{1,2} |}{c_{1} |{\cal E}_1 \cap E_1| + c_{1,2} |{\cal E}_1 \cap E_{1,2} |},
\end{align*}
where $C_2^{(1)}$ is replaced by $C_1$ due to symmetry.
\end{definition}

\begin{corrollary}\label{coro1}
The PoS can be arbitrary large, and it is lower bounded by 1, i.e., $PoS\geq 1$.
\end{corrollary}

\begin{proof}
Based on Proposition \ref{prop:order}, we can verify that in the symmetric scenarios with $n_1 (k+1) c_{1,2} < 1$, $C_1 = 0$ and $C_2 = n_1(k+1)c_{1,2}$, and hence $PoS$ can be infinite. For generally cases, we assume that $w$ is an SPE of the original game. Then, $w$ is also a feasible strategy profile for all three players if operator 2 plays first. Thus, we know that the utility of operator $2$ for acting first $C_2^{(1)}$ can at least achieve at a value $C_2$ resulting from the SPE $w$ of original game. Therefore, we conclude that $C_2/C_2^{(1)}\geq 1$, leading to $PoS\geq 1$.
\end{proof}

Corollary \ref{coro1} states that playing first is beneficial for the operator, and in certain cases the operator's optimal strategy becomes null if moving first in the game. Knowing this fact, the operator who has power to determine the order of players will prefer to move first by acting as a leader in this dynamic game to gain more benefits.

\subsection{More Threats Induce Collaboration}\label{collaboration_results}

In the adversarial network design, cyber attacks induce more collaboration (fair share of the costs) between two network operators. Specifically, when $k$ is small, operator $2$ may be the only player in creating the whole connected network. As $k$ increases (and as the number of required links increases), the cost for creating a connect infrastructure network increases. Then, at some point of $k$ when $(n_1+n_2) (k+1) > 1$, only network operator 2 cannot create a secure infrastructure network. Then, a set of new communication links is required from operator $1$ for the network to be $k$-connected. This phenomenon can be summarized as: a higher level of threats leads to collaboration between operators. Note that in some scenarios, the cost of operator 1 can be larger than that of operator 2 in network formation game. We illustrate these results using the following example.



\begin{example}
Let $n_1 = n_2 = 3$, $c_1 = c_{1,2}=0.19$ and $c_2 = c_{2,1}=0.21$.
Then, for $k=1$, the cost of operator 1 is $C_1 = 2c_1 = 0.38$, and the cost of operator 2 is $C_2 = 4c_2 = 0.84$. When $k=2$, the costs of operators 1 and 2 become $C_1 = 5c_1 = 0.95$ and $C_2 = 4 c_2 = 0.84$, respectively, where operator 1 costs more in creating a secure network than operator 2. In addition, as $k$ increases, operator 1 collaborates more with operator 2 to construct a $k$-resistant two-layer network. 
\end{example}

\begin{spacing}{1.7}    
 \begin{algorithm*}[!h]
 \caption{Algorithm to construct a secure interdependent IoT-enabled infrastructure network}\label{algo}
 	\DontPrintSemicolon
\KwIn{Parameters $c_1$, $c_2$, $c_{1,2}$, $c_{2,1}$, and odd $n_1$, $k$}
\KwOut{Network created by operator $1$ and $2$ at the SPE}
Let $e_1=0$ \tcp*{Number of links for operator $1$ to build}
Let $nb_{11} = (n_1-1) (k+1)/2$  \tcp*{The maximal number of links in $E_1$}
Compute $C_2(e_1) = n_1 (k+1) c_{2,1}$  \tcp*{Cost of operator $2$}
\While{$C_2(e_1) \geq 1$ \tcp*{Increase $e_1$ until $C_2$ becomes lower than $1$}}
{
 $e_1 = e_1+1$\;
\If{$e_1 \leq nb_{11}$}{
 $C_2(e_1) = C_2(e_1) - 2 c_{2,1} + c_2$} \Else
{$C_2(e_1) = C_2(e_1) - c_{2,1}$}
}

Let $e_{12} = \max(0, e_1-nb_{11})$;\ $e_{11} = e_1-e_{12}$
\tcp*{Number of links in $E_{1,2}$ and $E_{1}$ created by operator $1$}

Let $C_1(e_1) = e_{12} c_{1,2} + e_{11} c_1$\tcp*{Compute the resulting cost of operator $1$}
\If{$C_1(e_1) \geq 1$ \tcp*{The SPE is $((\emptyset, \emptyset), (\emptyset, \emptyset), \emptyset)$.}}
{exit(0)
}
Set $m = \floor*{\frac{e_{11}}{n_1-1}}$ \tcp*{Number of Hamiltonian cycles in $E_1$}

\For {$i=1$ to $m(n_1-1)$ \tcp*{Create $(n_1, m)$-Harary network without links $(1,i)$ with $2 \leq i \leq m$}}
{
  Let $j, \ell \geq 0$, such that $i-1 = \ell (n_1-1)+j$ and $j < n_1-1$\;
  $j = j+2$; $\ell = \ell+1$; \;
  Create links $(j; (j+\ell) \mod n_1)$ and $(j+n_1; ((j+\ell) \mod n_1) + n_1)$
}

\For {$i=1$ to $m$ \tcp*{Create links in $E_{1,2}$ such that $m$ Hamiltonian cycles are formed}}
{
  Create links $(1; n_1+1+i)$ and $(n_1+1; 1+i)$
}

Let $z = e_{11}-m(n_1-1)$ \tcp*{Number of extra links in $E_1$}
\If {$z>0$}
{
  $j=1$\;
  \For {$\ell =1$ to $z$}
  {
    $p = (j + (n_1+1)/2) \mod n_1$; \;
    Create link $(j;p)$ and $(j+n_1; p+n_1)$; \;
    Set $j=p$
  }
  $j=1$\;
  \For {$\ell =z+1$ to $n_1-1$}
  {
    $p = (j + (n_1+1)/2) \mod n_1$; \;
    Create link $(j;p+n_1)$ and $(j+n_1; p)$; \;
    $j=p$
  }
  Create link $(1; n_1+1)$ and $((n_1+1)/2; n_1 + (n_1+1)/2)$
}

Let $m_{12} = (k+1)/2-m$ \tcp*{Remaining Hamilton cycles in $E_{1,2}$} \quad \lIf {$z>0$}{$m_{12} = m_{12}-1$}

\For {$\ell = 1$ to $m_{12}$}
{
\For {$p=1$ to $n_1$}
{
Create link $(p; (p+m+\ell+1) \mod n_1 +n_1)$ and $(p+m + \ell+1) \mod n_1 + n_1; (p+1) \mod n_1)$\;
}}

\end{algorithm*}

\end{spacing}

\section{Design and Analysis of Algorithm for Secure Interdependent  Network Construction}\label{construct_alg}

With the obtained SPE in Section \ref{backward_induction}, the next critical step is to construct the secure interdependent IoT-enabled infrastructure networks. In this section, we present an algorithm to achieve this goal. For clarity purposes, we consider the scenario that $n_1 = n_2$. Further, we suppose that $n_1$ and $k$ are odd numbers. Knowing that in general mission-critical infrastructure networks, creating a link between two networks can be more difficult than creating a link within a network and thus induces a higher cost, we then have the conditions  $c_1 \leq c_{1,2}$ and $c_2 \leq c_{2,1}$ during the algorithm design.

\subsection{Network Construction Algorithm Design}
To construct a $k$-resistant IoT-enabled infrastructure network, operators $1$ and $2$ need to jointly create a network that has at least $\ceil*{\frac{n(k+1)}{2}} = n_1(k+1)$ links. This constitutes a lower bound on the number of links created (in a non-null strategy). Since $k$ is odd, by using Menger's theorem \cite{Menger}, we propose a construction using a superposition of exactly $\frac{k+1}{2}$ independent Hamiltonian cycles. The algorithm for the network construction is depicted in Algorithm \ref{algo}, and for clarity, we decompose it into 4 stages as follows.

\paragraph*{\textbf{Stage} $\mathbf{1}$} (Algorithm, line 1--15) In this stage, we determine the optimal values of $|{\cal E}_1^1|$ and $|{\cal E}_1^{1,2}|$. 

Denote $e_1 = |{\cal E}_1^1 \cup {\cal E}_1^{1,2}|$. 
For a given $e_1$, note that the cost of operator $1$ is minimized when $|{\cal E}_{1,2}|$ is minimized (since $c_1 \leq c_{1,2}$). 
Further, for each node of ${\cal N}_1$ whose degree is $d$ (with $d \leq k+1$), operator $2$ needs to create at least $k+1-d$ links in $E_{1,2}$. Note that any link of $E_1$ increases the degree of $2$ nodes in ${\cal N}_1$ by $1$, while any link of $E_{1,2}$ increases the degree of only $1$ node in ${\cal N}_1$ by $1$. Thus, 
each link created by operator $1$ in $E_1$ allows to decrease $|{\cal E}_2^{2,1}|$ by $2$, while each link created by operator $1$ in $E_{1,2}$ allows to decrease $|{\cal E}_2^{2,1}|$ by only $1$. Furthermore, for a given value of $e_1$, the cost of the second operator is minimized when $|{\cal E}_2^{2,1}|$ is minimized (since $c_2 \leq c_{2,1}$, as long as the sum of degrees of nodes in ${\cal N}_1$ is less than $n_1(k+1)-(k+1)$ since $k+1$ links are required in $E_{1,2}$).

Thus, for a given value of $e_1$,
both operators' costs are minimized when  $|{\cal E}_1^{1,2}|$ is minimized, that is when operator $1$ uses as many links between nodes of ${\cal N}_1$ as possible, as long as the sum of degrees of the nodes in ${\cal N}_1$ is less than $(n_1-1)(k+1)$.
Thus, for a given $e_1$,
$$e_{11} = |{\cal E}_{1}^1| = 
\left\{ 
\begin{array}{ll}
e_1 & \text{if } 2e_1 \leq (n_1-1)(k+1),\\
(n_1-1)(k+1)/2 & \text{otherwise.}
\end{array}\right.
$$

Thus, operator $1$ chooses the minimal value of $e_1$ and a set of links such that operator $2$ can
construct a $k$-resistant network with a cost lower than $1$.
Then, operator $1$ computes its own resulting cost. If it is higher than $1$, then no links are created and the SPE is $((\emptyset, \emptyset), (\emptyset, \emptyset), \emptyset)$. Otherwise, a network with $e_1$ links for operator $1$ and $n_1(k+1)-e_1$ links for operator $2$ is created. 

\paragraph*{\textbf{Stage} $\mathbf{2}$} (Algorithm, line 16--24) In this stage, we form $m$ independent Hamiltonian cycles with $m = \floor*{\frac{e_{11}}{n_1-1}}$.

First, operator $1$ creates 
links in $E_1$ in a similar manner as in Harary~\cite{harary}. That is, it first creates links between nodes $i$ and $j$ such that $(|i-j| \mod n_1) = 1$, and then $(|i-j| \mod n_1) = 2$, etc. From~\cite{jude_2016}, we know that a $2m$-Harary network contains exactly $m$ independent Hamilton cycles of ${\cal N}_1$, that is cycles that go through all $n_1$ nodes and such that no link is used more than once. Further, \cite{jude_2016} shows that there exists a construction such that links $(1;2)$, $(1;3)$, ..., $(1; m-1)$ all belong to different cycles. Thus, we remove those links from our construction and build all other links of the Harary network. 
We further construct $m(n_1-1)$ links in $E_2$ which are symmetric to as those in $E_1$. Note that these links in $E_2$ are created by operator 2. Hence, this stage creates $2m(n_1-1)$ links.

Further, for $1 \leq i \leq m$, 
by constructing two links, one between nodes $1$ and $n_1+i+1$ and one between nodes $n_1+1$ and $i+1$, we form a Hamiltonian cycle between all nodes in ${\cal N}_1 \cup {\cal N}_2$. Note that all $m$ different cycles use independent links. This further creates $2m$ links in ${E}_{1,2}$.

\paragraph*{\textbf{Stage} $\mathbf{3}$} (Algorithm, line 25--37) In the case where $e_{11} > m(n_1-1)$, then operator $1$ still needs to create $z = e_{11}-m(n_1-1)$ links in $E_1$. 

In that case, we create an additional Hamiltonian cycle in the following manner. Starting from node $1$, we consider the sequence $i_1; i_2; ... ; i_{n_1}$ with $i_1 = 1$ and $i_{j+1} = (i_{j}+ (n_1+1)/2) \mod n_1$. Since $n_1$ is odd, then for all $1 \leq j, \ell \leq  n_1$ and $j\neq \ell$, we have $i_j \neq i_\ell$ or in other words the sequence $i_1; i_2; ... ; i_{n_1}$ defines a permutation of indices $1, ..., n_1$. We then consider the following construction: for $j \leq z$, we construct the links $(i_j; i_{j+1})$ and $(i_{j}+n_1; i_{j+1}+n_1)$ and for $z < j < n_1$, we construct the links $(i_j; i_{j+1}+n_1)$ and $(i_{j}+n_1; i_{j+1})$. This defines $2$ sequences, and each one contains exactly $n_1$ nodes. By adding links $(1; n_1+1)$ and $((n_1+1)/2; n_1 + (n_1+1)/2)$, we create a full Hamiltonian cycle. Note that none of the links used previously have been created since $m < (k+1)/2$. This stage creates exactly either $0$ link or $2n_1 = n$ links among which $z$ links are in $E_1$, $z$ links are in $E_2$, and $n-2z$ links are in $E_{1,2}$. 

\paragraph*{\textbf{Stage} $\mathbf{4}$} (Algorithm, line 38--42) In total, either $m$ or $m+1$ Hamiltonian cycles have been created and $e_{11}$ links have been used. We thus create the remaining $m_{12}$ Hamiltonian cycles with links exclusively in $E_{1,2}$ that have not been created in the previous stages. Here, $m_{12}=(k+1)/2-m$ if $z\leq 0$ and $m_{12}=(k+1)/2-m-1$ otherwise.  A possible solution for $k<n_1$ is as follows. For all $t$ that satisfy $m < t < (k+1)/2$, we construct a Hamiltonian cycle following this pattern:
for any $1 \leq i \leq n_1$, we create links $(i; (i+t \mod n_1) +n_1)$ and $((i+t) \mod n_1 +n_1; (i+1)\mod n_1)$ in the network.

The above 4 stages of construction yield an equilibrium two-layer secure IoT-enabled infrastructure network.

\begin{figure}[!t]
\begin{centering}
\includegraphics[width=1\columnwidth]{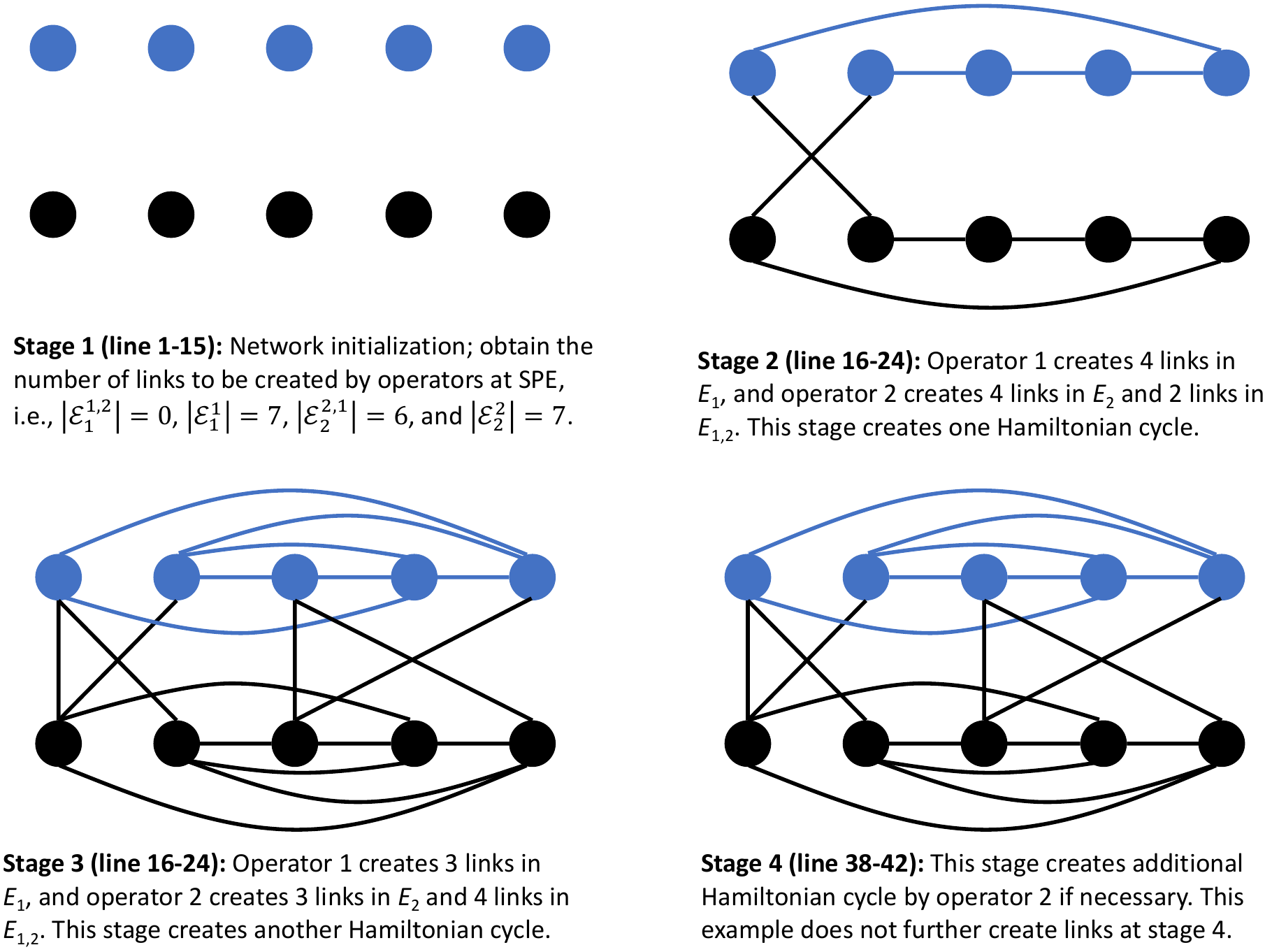}
\par\end{centering}     
\caption{\label{example_5_node}
Illustration of the secure network construction process using Algorithm \ref{algo}. Both subnetworks contain 5 nodes. Blue-colored links are created by operator 1, and black-colored links are established by operator 2. The network constructed at SPE is resistant to any 3 link attacks.  Stage 1 of the algorithm yields the number of links that each operator needs to create at the SPE. Stage 2-4 constructs the secure network that satisfies the requirements. In this example, one Hamiltonian cycle is created at both stages 2 and 3, while no further links are constructed at stage 4, as network after stage 3 is resistant to 3 link attacks.}
\end{figure}

\textit{An Illustrative Example:}
We next present an example to illustrate how the developed algorithm works. The parameters are summarized as follows: $n_1=n_2=5$, $k=3$, $c_1=c_2 = \frac{1}{20}$, and $c_{1,2} = c_{2,1} = \frac{1}{10}$. Fig. \ref{example_5_node} depicts the detailed steps of the secure network construction process. In stage 1 of the algorithm (line $1-15$), we obtain the following: $e_{12} =|{\cal E}_1^{1,2} | = 0$, $e_{11} = |{\cal E}_1^1| = 7$, which means that operator 1 creates 0 link in $E_{1,2}$ and 7 links in $E_1^1$ at SPE. Furthermore, operator 2 creates  6 links in $E_{1,2}$ and 7 links in $E_2$ at SPE. Hence, the net utilities of operators 1 and 2 at SPE are $\frac{13}{20}$ and $\frac{1}{20}$, respectively. Next, stage 2-4 of the algorithm creates such a secure infrastructure network which is resistant to $k=3$ link removal attacks. Specifically, at stage 2 (line $16-24$), two network operators create an independent Hamiltonian cycle, where operator 1 creates 4 links in $E_1$ and operator 2 creates 4 links in $E_2$ and 2 links in $E_{1,2}$. After this stage, the network is resistant to 1 link attack.  At stage 3 (line $25-37$), two operators create an additional Hamiltonian cycle in our example, after which each node is of degree 4 and thus the network is resistant to 3 link attacks. At this stage, operator 1 further creates 3 links in $E_1$ and operator 2 creates 3 links in $E_2$ and 4 links in $E_{1,2}$. Both operators at stage 4 of the algorithm (line 38-42) do not construct links, as there is no remaining Hamilton cycle to be further created in $E_{1,2}$. This example shows that the algorithm can successfully guide the design of secure interdependent networks.

\subsection{Complexity and Scalability Analysis}
We proceed to analyze the property of the designed algorithm in this section. The first one is the complexity. In the construction of secure networks based on Algorithm \ref{algo}, we allocate the links in two networks sequentially. Depending on the system parameters, the total number of allocated links is achieved between $\ceil*{\frac{(n_1+n_2)(k+1)}{2}}$ and $\ceil*{\frac{n_1(k+1)}{2}}+\ceil*{\frac{n_2(k+1)}{2}}+k+1$. Note that $\ceil*{\frac{(n_1+n_2)(k+1)}{2}}$ is a lower bound in the sense that it is the minimum number of links for a network containing $n_1+n_2$ nodes being resistant to $k$ attacks. Similarly, $\ceil*{\frac{n_1(k+1)}{2}}+\ceil*{\frac{n_2(k+1)}{2}}+k+1$ is an upper bound for a secure network where each subnetwork is resistant to $k$ attacks and another $k+1$ links are allocated between two subnetworks. Therefore, the complexity of Algorithm \ref{algo} is linear in the number of nodes and the number of attacks, i.e., $O(nk)$, where $n=n_1+n_2$. 

We next analyze the scalability of the algorithm. Scalability is critical when the system parameter changes and the equilibrium network becomes different and needs to be reconfigured. In our algorithm, the main procedure in designing the secure network is constructing Hamiltonian cycles sequentially in both subnetworks and between two layers. This construction pattern is highly scalable in the sense that when the number of nodes in two subnetworks or the number of attacks changes, the solution infrastructure network at SPE adopts with a similar topology and only a small set of existing links need to be rewired. The scalability of the algorithm will be further justified using case studies in Section \ref{number_agents}.

\section{Case Studies} \label{case_study}
In this section, we use case studies of IoBT to illustrate the optimal design principles of secure  IoT-enabled networks with heterogeneous components. In a battlefield scenario, the unmanned ground vehicles (UGV) and unmanned aerial vehicles (UAV) execute missions together. To enhance the information transmission quality and situational awareness of each agent in the battlefield, a secure and reliable communication network resistant to malicious attacks is critical. Note that the results in this section are also applicable to other mission-critical heterogeneous IoT network applications.

\subsection{Secure Interdependent Network Design at SPE}\label{base_case_study}
In the following case studies, we consider $n_1=9$ UAVs and $n_2=9$ UGVs in the two-layer IoBT network. The normalized costs of creating different types of links are as follows: $c_1 = \frac{1}{30}$, $c_{1,2} = \frac{1}{20}$, $c_2=\frac{1}{45}$, and $c_{2,1} = \frac{2}{45}$. Here, we can see that the cost of creating intralinks across two layers is more than that of creating interlinks within the network itself. In addition, the normalized unit cost of attack is $c_A = \frac{1}{3}$, and hence the attacker can compromise at most $k=3$ links in the network. The above normalized costs can be transformed to their nominal monetary costs when the basis ratio is determined. Based on Lemma \ref{lem:4}, we obtain that, at SPE, the UAV network operator 1 creates 10 interlinks within its own network, and the UGV network operator 2 formulates 10 interlinks as well as 16 intralinks between two layers in the IoBT. Therefore, the equilibrium payoffs for operators 1 and 2 are $U_1^* = \frac{5}{6}$ and $U_2^* = \frac{1}{15}$, respectively. Note that the equilibrium IoBT network is a 3-connected network, and thus the attacker is incapable of disconnecting the system even with his best effort. By using the designed Algorithm \ref{algo}, we construct the solution IoBT network resistant to $3$ attacks and the result is shown in Fig. \ref{case_A}. We can verify that operator 2 not only allocates link resources in his own UGV network but also in the places connecting two interdependent layers. Furthermore, each node in the network is of degree 4, and the network is resistant to anticipated attacks.

\subsection{Impact of the Number of Attacks}
We next investigate the impact of the number of attacks on the adversarial IoBT network formation. Varying $k$ captures the operators' belief on the attacker's incentives. A larger $k$ indicates that the operators anticipating a higher level of threats, and thus the designed IoBT network needs to be more resistant.  In the following, the link creation costs are the same as those in the previous case study. We vary the attacker's capability $k$, and the obtained results are shown in Fig. \ref{case_B}. Fig. \ref{case_B_1} illustrates the number of formed links of network operators 1 and 2 at the equilibrium IoBT configuration. When the attacker can compromise less than 2 links, the UGV network operator creates sufficient interlinks that connect UAVs and UGVs. Therefore, the utility of UAV network operator is 1. As the number of attacks increases, operator 1 begins to contribute to the network defense because operator 2 alone cannot secure the network with a positive payoff. For $2\leq k<7$, operator 1 allocates link resources only within the UAV network. In comparison, operator 2 creates fewer intralinks and allocates more resources in its own UGV network as the cyber threats increase. In addition, when the number of attacks exceeds a certain level, i.e., $k\geq 7$ in this case, both network operators will cease to protect the network, and the corresponding SPE is a null strategy which satisfies the second condition in Lemma \ref{lem:5}. Fig. \ref{case_B_2} shows the utilities of two operators at the equilibrium IoBT network. The operator 1's payoff decreases as $k$ grows. Interestingly, in the regime where the UAV network operator contributes to the secure IoBT network, i.e.,  $2\leq k \leq 6$, the utility of UGV network operator remains the same which corresponds to the maximum effort that operator 2 can use. Based on this case study, we can conclude that higher threat levels induce more collaborations between two network operators. 

\begin{figure}[!t]
\begin{centering}
\includegraphics[width=0.9\columnwidth]{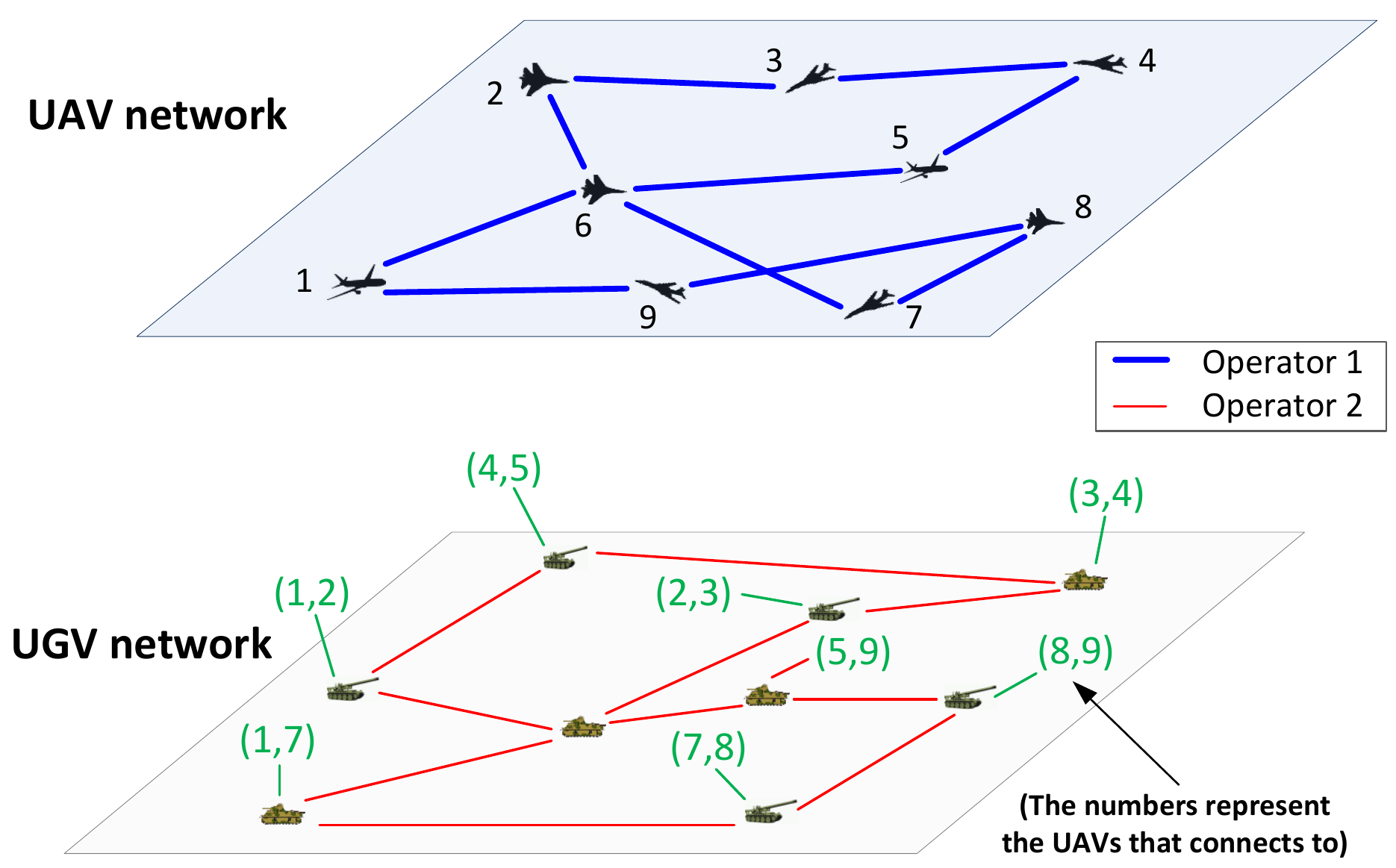}
\par\end{centering}     
\caption{\label{case_A}
The 3-connected equilibrium IoBT network.}
\end{figure}

\begin{figure}[!t]
  \centering
  \subfloat[Number of links in the equilibrium network]{
    \includegraphics[width=0.85\columnwidth]{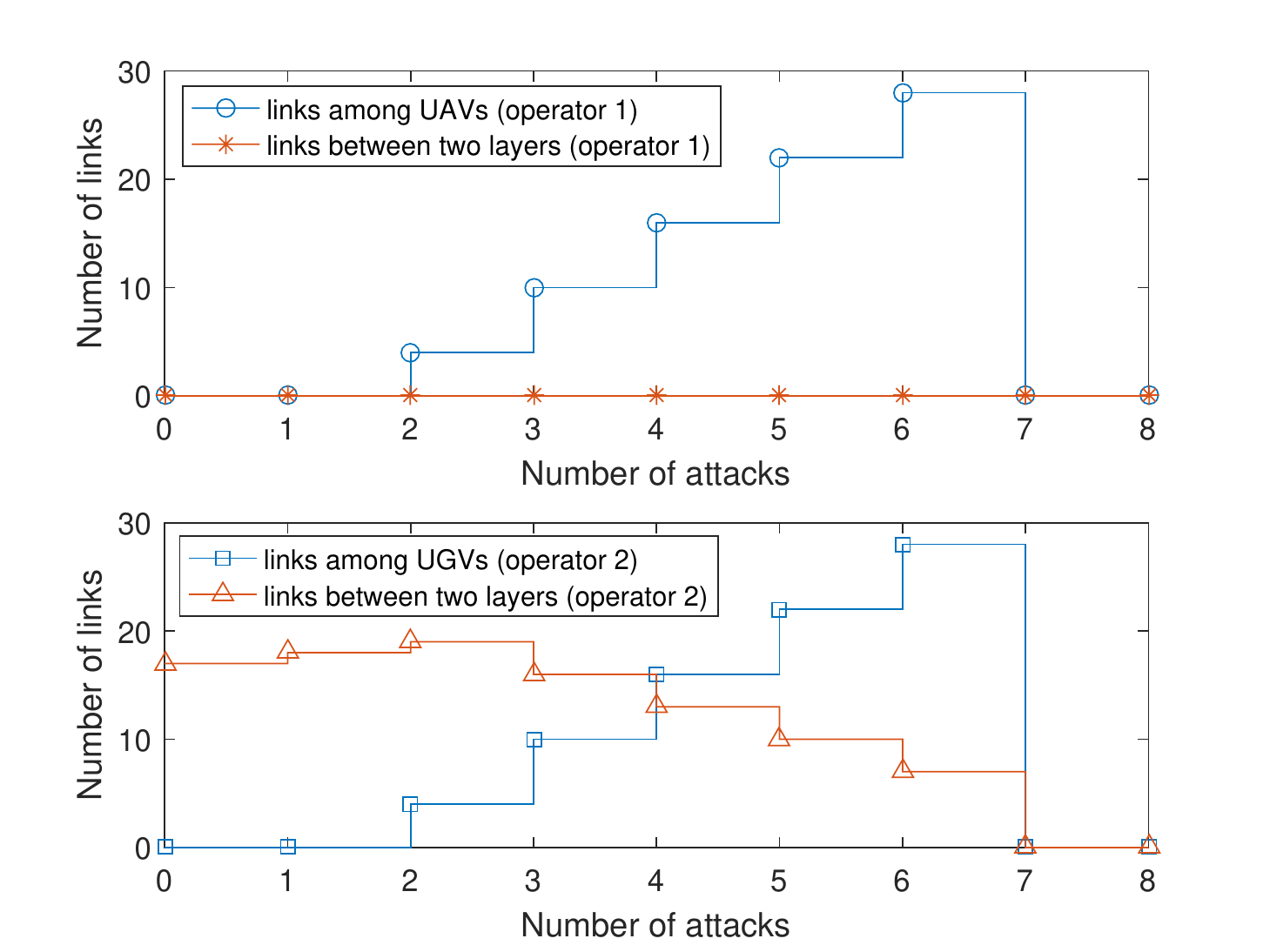}\label{case_B_1}}\\ 
	 \subfloat[Equilibrium utility of operators]{
    \includegraphics[width=0.85\columnwidth]{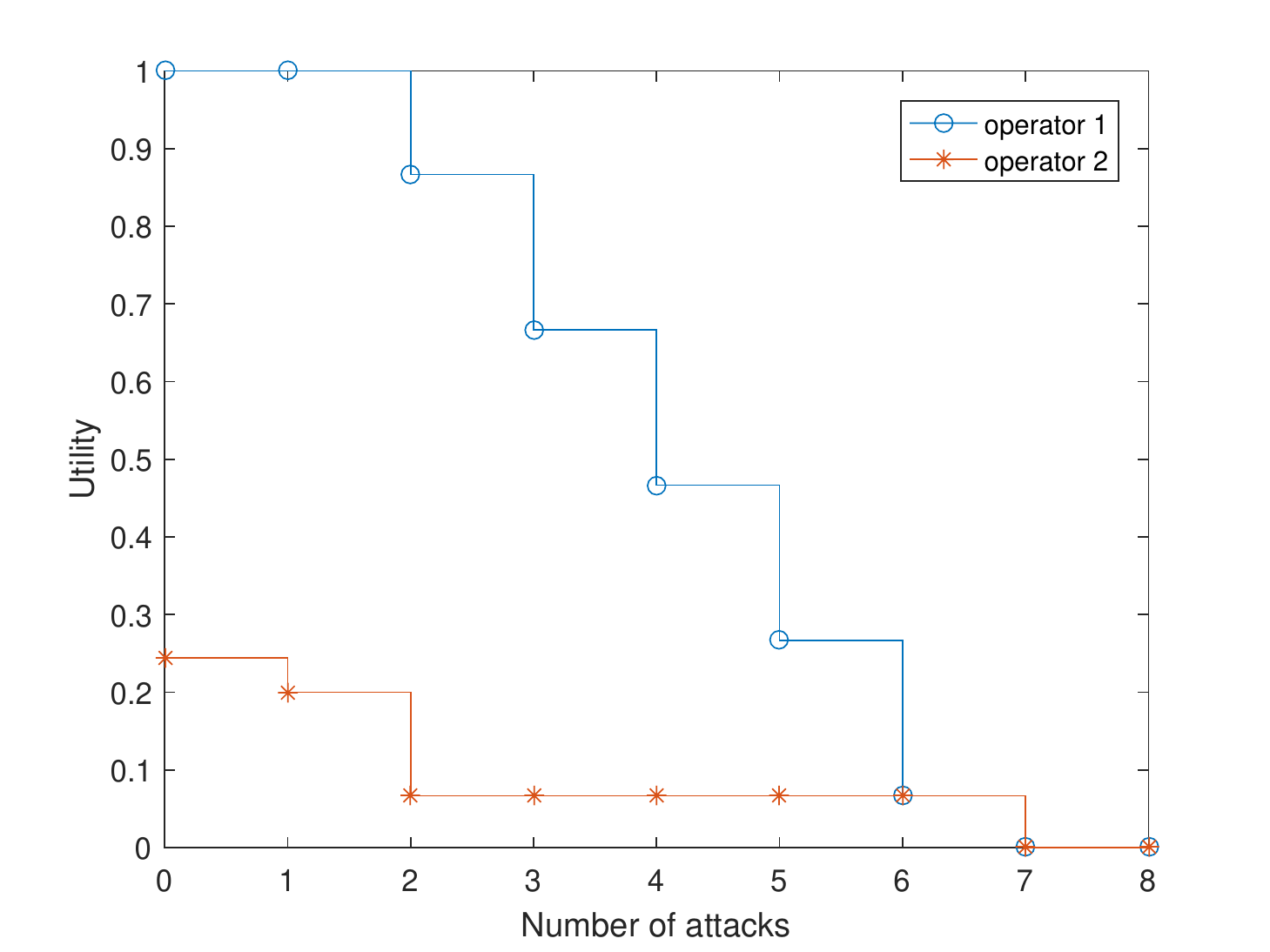}\label{case_B_2}}
  \caption[]{(a) shows the number of different types of links that operators 1 and 2 should create at the equilibrium IoBT network. (b) is the corresponding utility of two operators. }
  \label{case_B}

\end{figure}

\subsection{Impact of the Number of Agents}\label{number_agents}
To illustrate the scalability of the designed algorithm, we next present the results with different  number of agents in the network. This scenario is practical in mission-critical applications as the number of agents may change over time. The system parameters are the same as those in Section \ref{base_case_study} except that $n_1$ and $n_2$ are different (agents leaving or joining the battlefield). Note that as $n_1$ and $n_2$ change, the IoBT configurations at SPE are also different. We present the results when some agents in two subnetworks leave the battlefield, i.e., $n_1=8$ and $n_2 = 8$. At the SPE, operator 1 creates 7 interlinks within its network, and operator 2 formulates 7 interlinks and 18 intralinks between two layers in the IoBT. The equilibrium payoffs for operators 1 and 2 are $U_1^* = \frac{23}{30}$ and $U_2^* = \frac{2}{45}$, respectively. Figure \ref{case_C} shows the constructed network using Algorithm \ref{algo}. We can see that comparing with the result in Fig. \ref{case_A}, the links associated with those leaving agents are removed and link rewiring only happens to a small set of nodes originally connected with leaving agents. The major network configuration stays unchanged which demonstrates the scalability of the algorithm.

\begin{figure}[!t]
\begin{centering}
\includegraphics[width=0.9\columnwidth]{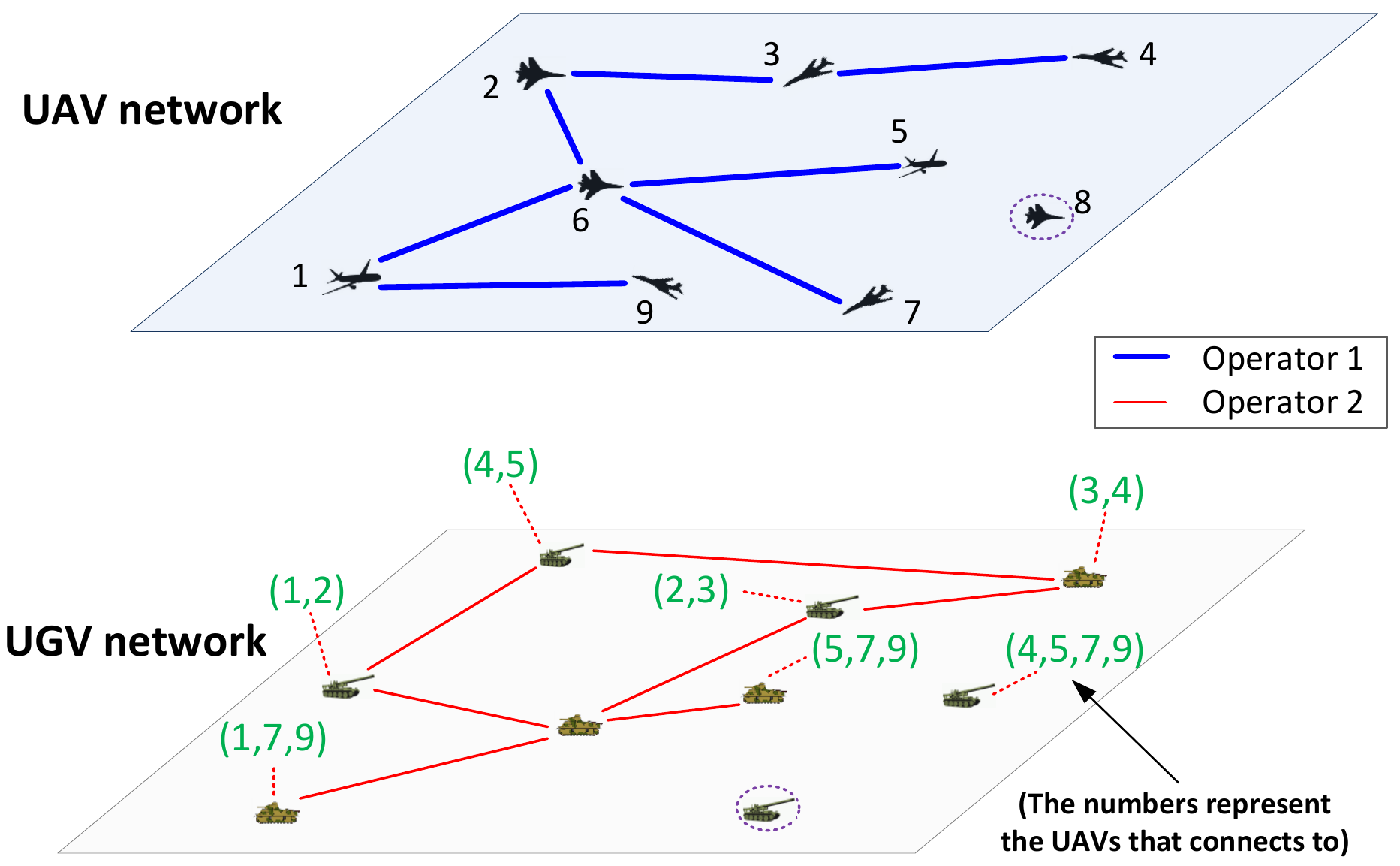}
\par\end{centering}   
\caption{\label{case_C}
The 3-connected equilibrium IoBT network with two agents leaving the field. Comparing with Fig. \ref{case_A}, only a small subset of links need to be rewired which shows the scalability of the algorithm.}
\end{figure}

\section{Conclusion}\label{conclusion}
In this paper, we have investigated the adversarial network design for the interdependent IoT-enabled infrastructures. To secure the heterogeneous components in the infrastructure networks, we have formulated a three-player three-stage network formation game where two network operators aim to keep the network connected in the presence of attacks. The subgame perfect Nash equilibrium (SPE) of the game has been shown to be an empty set when the number of communication links that the attacker can compromise exceeds a threshold, or the link creations are too costly for the operators. The price of anarchy (PoA), i.e., the ratio of network formation costs between the SPE and team optimal strategies, can be unbounded. Also, we have identified cases where SPE solution is as efficient as the team optimal one, yielding PoA equaling 1. Furthermore, with a higher threat level, we have shown that two network operators are more willing to collaborate to defend against attacks, since one operator alone cannot completely mitigate the threats with a limited amount of link resources. The future work would extend the dynamic game model and investigate the case in which network operators can create secure links and insecure links with distinct costs. Under such modeling, the adversary also has heterogeneous attack costs to different types of communication links.

\bibliographystyle{IEEEtran}
\bibliography{IEEEabrv,references}

\end{document}